\documentclass[3p,11pt,preprint]{elsarticle}
\usepackage{amsfonts}
\usepackage{amsmath}
\usepackage{amssymb}
\usepackage{mathtools}
\usepackage{bbm}

\biboptions{authoryear}
\usepackage{tikz}
\usetikzlibrary{matrix}
\usetikzlibrary{decorations.pathreplacing}

\newdefinition{definition}{Definition}
\newdefinition{example}{Example}
\newtheorem{theorem}{Theorem}
\newtheorem{lemma}{Lemma}
\newtheorem{corollary}{Corollary}
\newproof{proof}{Proof}

\let\set\mathbbm
\def\k{\set D}
\newcommand{\qk}{\set{K}}
\newcommand{\D}{x}
\newcommand{\alg}{\k[\D;\sigma,\delta]}
\newcommand{\qalg}{\qk[\D;\sigma,\delta]}
\newcommand{\lc}{\operatorname{lc}}
\newcommand{\gcrd}{\operatorname{gcrd}}
\newcommand{\syl}{\operatorname{Syl}}
\newcommand{\cont}{\operatorname{cont}}
\newcommand{\pquo}{\operatorname{pquo}}
\newcommand{\prem}{\operatorname{prem}}
\newcommand{\pres}{\operatorname{sres}}
\newcommand{\prs}[3]{(#1_i)_{i\in \{#2,\dots,#3\}}}
\newcommand{\cc}{\operatorname{tc}}

\journal{Journal of Symbolic Computation}

\begin{document}
\begin{frontmatter}
  \title{Improved Polynomial Remainder Sequences for Ore~Polynomials}

  \author{Maximilian Jaroschek\fnref{fn1}}
  \ead{mjarosch@risc.jku.at}
  \address{Research Institute for Symbolic Computation\\ Johannes Kepler University\\ A4040 Linz, Austria}
  \fntext[fn1]{Supported by the Austrian Science Fund (FWF) grant Y464-N18}

\begin{abstract}
Polynomial remainder sequences contain the intermediate results of the Euclidean algorithm when applied to (non-)commutative polynomials. The running time of the algorithm is dependent on the size of the coefficients of the remainders. Different ways have been studied to make these as small as possible. The subresultant sequence of two polynomials is a polynomial remainder sequence in which the size of the coefficients is optimal in the generic case, but when taking the input from applications, the coefficients are often larger than necessary. We generalize two improvements of the subresultant sequence to Ore polynomials and derive a new bound for the minimal coefficient size. Our approach also yields a new proof for the results in the commutative case, providing a new point of view on the origin of the extraneous factors of the coefficients.
\end{abstract}


\end{frontmatter}

\section{Introduction}

When given a system of differential equations, one might be interested in finding the common solutions of these equations. In order to do so, one can compute another differential equation whose solution space is the intersection of the solution spaces of the equations in the original system. One way to do this is to translate the equations into operators and use the Euclidean algorithm to compute their greatest common right divisor. The solution space of the greatest common right divisor then consists of the desired elements. 

Similarly, given a sequence of numbers $(t_n)_{n\in\{0,1,\dots\}}$ that satisfies two different recurrence equations, the Euclidean algorithm is used in applications to find a reasonable candidate for the least order equation of which $(t_n)_{n\in\{0,1,\dots\}}$ is a solution.

Carrying out Euclid's algorithm applied to two polynomials over a domain $\k$ usually requires~a~prediction of the denominators that might appear in the coefficients of the remainders in order to bypass costly computations in the quotient field of $\k$. While such a prediction can be done easily, the growth of the coefficients of the remainders can be tremendous, which might result in an unnecessary high running time. This can be avoided by dividing out possible content of the remainders to make their coefficients as small as possible. For commutative polynomials as well as for non-commutative operators, different ways have been extensively studied to find factors of the content in the sequence of remainders without computing the GCD of the coefficients of each element of the sequence. Most notably in this respect are subresultant sequences, where the growth of the coefficients can be reduced from exponential to linear in the number of reduction steps in the Euclidean algorithm. When taking generic, randomly 
generated input, the coefficient size in the subresultant sequence is usually optimal, but when taking the input from applications in e.g. combinatorics or physics, the remainders still have non-trivial content in many cases. 

For commutative polynomials, some ways are known to improve on subresultants. In this article we generalize two of these results to Ore polynomials and we also give a new proof for the commutative case that is based on the structure of subresultants as matrix determinants. Furthermore, we use these results to derive a new bound for the coefficient size of the content-free remainders.

In Section~\ref{prelimsec} the basic notions of Ore polynomial rings are stated. A precise definition and examples of polynomial remainder sequences are given in Section~\ref{prssec} and further details on the subresultant sequence are then presented in Section~\ref{sressec}. The main results of this article can be found in Sections~\ref{contentsec}~and~\ref{improvsec}, where we first describe how additional content in the subresultant sequence can emerge and then use these results to improve on the Euclidean algorithm and to get~a~new bound for the size of the coefficients.

\section{Preliminaries}
\label{prelimsec}

The algebraic framework for different kinds of operators that we consider here are Ore polynomial rings, which were introduced by \O ystein Ore in the 1930's. We provide an overview of some basic facts that suffice our needs and that can be found in \citet{ore} and \citet{bronstein}.
\begin{definition}
Let $\k$ be a commutative domain, $\k[\D]$ the set of univariate polynomials over $\k$ and let $\sigma\colon \k\rightarrow \k$ be an injective endomorphism.
\begin{enumerate}
 \item  A map $\delta\colon \k\rightarrow \k$ is called pseudo-derivation w.r.t. $\sigma$, if for any $a,b\in \k$ 
\begin{displaymath}
 \delta(a+b) = \delta(a) + \delta(b)\text{\quad and \quad} \delta(ab) = \sigma(a)\delta(b) + \delta(a)b.
\end{displaymath}
 \item Suppose that $\delta$ is a pseudo-derivation w.r.t. $\sigma$. We define the Ore polynomial ring $(\k[\D],+,\cdot)$ with componentwise addition and the unique distributive and associative extension of the multiplication rule
\begin{displaymath}
 \D a = \sigma(a)\D+\delta(a)\text{\qquad for any }a\in\k,
\end{displaymath}
 to arbitrary polynomials in $\k[\D]$. To clearly distinguish this ring from the standard polynomial ring over $\k$, we denote it by $\alg$.
\end{enumerate}
\end{definition}

 Elements of an Ore polynomial ring are called operators and are denoted by capital letters. We refer to the leading coefficient of an operator $A$ as $\lc(A)$, to the coefficient of $\D^0$ in $A$ as $\cc(A)$ and to the polynomial degree of $A$ in $x$ as the order $d_A$ of $A$. 

\begin{example}
\label{oreex} Commonly used Ore polynomial rings are:
 \begin{enumerate}
  \item $\k[\D]=\k[\D;1,0]$, the ring of commutative polynomials over $\k$.
  \item $\set C(y)[D;1,\frac{d}{dy}]$, the ring of linear ordinary differential operators.
  \item If $s_n\colon\set C(n)\rightarrow \set C(n)$ is the forward shift in $n$, i.e. $s_n(a(n)) = a(n+1)$, then $\set C(n)[S;s_n,0]$ is the ring of linear ordinary recurrence operators.
  \item If $\sigma\colon\set C(q)(y)\rightarrow \set C(q)(y)$ is the $q$-shift in $y$, i.e. $\sigma(a(y)) = a(qy)$, then $\set C(q)(y)[J;\sigma,\frac{d}{dy}]$ is the ring of Jackson's $q$-derivative operators.
 \end{enumerate}
\end{example}

In this article, we consider the following situation: Let $\k$ be a Euclidean domain with degree function $\deg$ and let $\alg$ be an Ore polynomial ring where $\sigma$ is an automorphism. For any operator~$A\in\alg$, we define $\lVert A\rVert$ to be the maximal coefficient degree of $A$. The content $\cont(A)$ of $A$ is the greatest common divisor of all the coefficients of $A$ and it is defined to be $\lc(A)$ if $\k$ is a field. It is possible to extend $\alg$ to an Ore polynomial ring over the quotient field~$\qk$ of $\k$ by setting $\sigma(a^{-1}) = \sigma(a)^{-1}$ and~$\delta(a/b) = (b\delta(a)-a\delta(b))/(b\sigma(b))$ for $a,b\in\k,$ $b\neq 0$ (see \cite{zli2}, Proposition 2.2.1). We will denote this ring by $\qalg$ without making it explicit that the automorphism and the pseudo-derivation are extensions of the functions used in $\alg$. It is well known that for any two operators $A,B\in\qalg$, there exists a greatest common right divisor (GCRD) and it can be made unique (up to units in $\k$) by 
setting $\gcrd(A,B)$ to a nonzero $\qk$-
left multiple of any GCRD of $A$ and~$B$ that has coefficients in $\k$ but does not have any content in $\k$. 

Throughout this article, we let $A,B,G\in\alg$, $B\neq 0$ be such that $d_A\geq d_B$ and~$G$ is the GCRD of $A$ and~$B$.

\begin{definition}
 For $a\in\k$ and~$n\in\set N$, $\sigma^n(a)$ is obtained by applying $n$ times $\sigma$ to $a$ and~$\sigma^{-n}(a):=(\sigma^{-1})^n(a)$, where $\sigma^{-1}$ is the inverse map of $\sigma$. The $n$th $\sigma$-factorial of $a\in\k$ is defined as the product \[a^{[n]} := \prod\limits_{i=0}^{n-1}{\sigma^i(a)}.\]
\end{definition}

\section{Polynomial Remainder Sequen\-ces for Ore Polynomials}
\label{prssec}

The greatest common right divisor of $A$ and~$B$ can be computed by using the Euclidean algorithm. If we multiply any intermediate result that appears during the execution of the  algorithm by an element of $\qk\setminus\{0\}$, the final output will be a $\qk$-left multiple of $G$. This amount of freedom allows us to optimize the running time by choosing these factors appropriately. In order to be able to formulate improvements of this kind, the notion of polynomial remainder sequences has been introduced. Each element of such a sequence corresponds to a remainder computed in one iteration of the Euclidean algorithm.
\begin{definition}
 Let $\prs{R}{0}{\ell+1}$ and~$\prs{Q}{1}{\ell}$ 
%$\left(S_i\right)_{i\in\{0,\dots,\ell+1\}}$ and~$\left(T_i\right)_{i\in\{0,\dots,\ell+1\}}$ 
be sequences in $\qalg$, $\prs{d}{0}{\ell}$ a sequence in $\mathbb{N}$ and let $\prs{\alpha}{1}{\ell}$ and~$\prs{\beta}{1}{\ell}$ be sequences in $\qk$ such that
 \begin{alignat*}4
  & R_0    = A,\quad R_1 = B,\quad d_i    = d_{R_i},\\
  & \alpha_iR_{i-1}  = Q_iR_i + \beta_iR_{i+1},\quad d_{i+1} < d_i, 
%   & \phantom{\alpha^{(i)}}r^{(i)}   & = &\ s^{(i)}p + t^{(i)}q,\quad\deg s^{(i)} = n-d^{(i-1)},\text{ } \deg t^{(i)} = m-d^{(i-1)}\\
%   & \rlap{$\forall i\in\{0,\dots,\ell\}:R_i  \neq  0,\quad R_{\ell+1} = 0.$}
 \end{alignat*}
and all $R_i$ are nonzero except for $R_{\ell+1}$. We call the sequence $\prs{R}{0}{\ell+1}$ a polynomial remainder sequence (PRS) of $A$ and~$B$. 
% A PRS is called normal if $d_i = d_{i-1}-1$.
\end{definition}

A PRS of $A$ and~$B$ is uniquely determined by specifying the $\alpha_i$ and~$\beta_i$. Whenever we talk about~a~PRS $\prs{R}{0}{\ell+1}$, we allow ourselves to refer to the related sequences $\prs{Q}{1}{\ell}$, $\prs{d}{0}{\ell}$ etc. as in the above definition without explicitly introducing them. 

In order to efficiently compute $G$, one wants to make sure that all the remainders are elements of $\alg$ rather than $\qalg$. This can be achieved by choosing the $\alpha_i$ in a way such that the quotient of any two consecutive remainders has coefficients in $\k$. To this extent, for $1\leq i\leq\ell$ set $\alpha_i:=\lc(R_i)^{[d_{i-1}-d_i+1]}$ and division with remainder yields $Q_i$ and~$R_{i+1}$ in $\alg$ with:
\begin{alignat}2
\label{premformula}
 \alpha_i R_{i-1} = Q_i{R_i} + R_{i+1},\qquad d_{i+1}<d_i.
\end{alignat}
We call $\pquo(R_{i-1},R_i):=Q_i$ the pseudo-quotient of $R_{i-1}$ and~$R_i$ and~$\prem(R_{i-1},R_i):=R_{i+1}$ the pseudo-remainder of $R_{i-1}$ and~$R_i$.

The $\alpha_i$ are used to make sure that computations can be done in $\alg$ and the $\beta_i$ control the coefficient growth in a PRS. We want $\beta_i$ to contain as many factors of the content of $R_{i+1}$ as possible without much computational overhead needed to obtain these factors.

\begin{example}
\label{prsex}
 Set $\alpha_i =\lc(R_i)^{[d_{i-1}-d_i+1]}$ and
\begin{enumerate}
 \item $\beta_i = 1$. This is called the pseudo PRS of $A$ and~$B$. Here, no content will be divided out.
 \item $\beta_i = \cont(R_{i+1})$. This is called the primitive PRS of $A$ and~$B$. The coefficients of the remainders will be as small as possible, but it is necessary to compute the GCD of the coefficients of each remainder in order to get the $\beta_i$.
 \item \label{sresprs}The subresultant PRS of $A$ and~$B$ (see Section~\ref{sressec}) is given by
\begin{alignat*}2
 \beta_i & = \left\{\begin{array}{ll}
\displaystyle -\sigma(\psi_1)^{[d_0 - d_1]}, & \text{ if } i=1,\\
\displaystyle -\lc(R_{i-1})\sigma(\psi_i)^{[d_{i-1}-d_i]}, & \text{ if } 2\leq i\leq\ell,\end{array}\right.
\end{alignat*}
where
\[
\psi_i = \left\{\begin{array}{ll}
                 \displaystyle -1, & \text{ if } i = 1,\\
                 \displaystyle \frac{(-\lc(R_{i-1}))^{[d_{i-2}-d_{i-1}]}}{\sigma(\psi_{i-1})^{[d_{i-2}-d_{i-1}-1]}}, & \text{ if } 2\leq i\leq\ell.
                \end{array}\right.
\]
In this PRS, the content that is generated systematically by pseudo-remaindering will be cleared from the remainders.
\end{enumerate}
\end{example}

While in all of the above PRSs the remainders are elements of $\alg$, the degrees of the coefficients differ drastically, as illustrated in the following example. It can be shown that the degrees of the coefficients in the pseudo PRS grow exponentially with $i$, which renders this PRS practically useless. The growth in the subresultant and primitive PRS is linear in $i$.

\begin{example}
\label{mainex}
 Assume we are given a finite sequence of rational numbers that comes from a sequence $(t_n)_{n\in\{0,1,\dots\}}$ which admits a linear recurrence equation with polynomial coefficients. If the amount of data is sufficiently large, we are able to guess recurrence operators of some fixed order and maximal coefficient degree that annihilate $(t_n)_{n\in\{0,1,\dots\}}$, i.e. the operators applied to the sequence give zero. (For details on guessing and a Mathematica implementation of the method, see \cite{kauers}.)
 For example, consider 
 \[t_n = \sum_{k=0}^{n}{\binom{2n+4}{k}+(2n-k)! + k^3}.\]
 Given the first 300 terms of this sequence, we can find two operators $A$ and~$B$ in $\set Q[n][S;s_n,0]$ with $d_A=14$, $d_B=13$ and maximal coefficient degree $\lVert A\rVert=5$, $\lVert B\rVert=6$ resp. Both operators annihilate the given sequence, but none of them is of minimal order. To get an annihilating minimal order operator, we compute the GCRD of $A$ and~$B$ in $\set Q(n)[S;s_n,0]$. Table 1 shows the maximal coefficient degrees of the remainders for different PRSs of $A$ and~$B$.
\begin{center}
\begin{tabular}{l|c|c|c|c|c|c|c}
% \hline
PRS & $R_2$ & $R_3$ & $R_4$ & $R_5$ & $R_6$ & $R_7$ & $R_8$\\
\hline
 pseudo & 11 & 22 & 49 & 114 & 271 & 650 & 1565\\
\hline
 subresultant & 11 & 16 & 21 & 26 & 31 & 36 & 41\\
\hline
 primitive & 9 & 12 & 15 & 18 & 21 & 24 & 21\\
% \hline
\end{tabular}\\
\vspace{0.1cm}\scriptsize{Table 1: Maximal coefficient degrees for different PRSs.}
\end{center}
The example confirms that the degrees in the pseudo PRS grow exponentially, whereas the subresultant PRS and the primitive PRS show linear growth. At the same time, the degrees in the subresultant PRS are not as small as possible. This behavior is typical not only for this pair~$A$~and~$B$, but in general for operators coming from applications. For randomly generated operators, the subresultant PRS and the primitive PRS usually coincide. Our goal is to understand the difference between randomly generated input and the operators $A$ and~$B$ as above and to identify the source of some (and most often all) of the additional content in the subresultant PRS. To make use of this knowledge, we will then adjust the formulas for $\alpha_i$ and~$\beta_i$ from Example~\ref{prsex}.\ref{sresprs} so that we get a PRS with smaller degrees without having to compute the content of every remainder.
\end{example}

\section{Subresultant Theory for Ore Polynomials}
\label{sressec}

For commutative polynomials, the theory of subresultants was intensively studied by \cite{brown}, \cite{browntraub}, \cite{collins} and \cite{loos}. The main idea is to translate relations between the elements of a PRS like the B\'ezout relation or the (pseudo-)remainder formula into linear algebra. A central tool in this context is the Sylvester matrix, which, roughly speaking, contains the coefficients of all the monomial multiples of the input polynomials that are necessary to compute remainders of any possible degree. The remainders in the subresultant sequence turn out to be polynomials whose coefficients are determinants of certain submatrices of this matrix. \cite{zli} generalized these results to Ore polynomials. 

\begin{center}
\begin{tikzpicture}
{
\matrix (m) [matrix of math nodes,left delimiter=(\hspace{0.7cm},right delimiter={)}]
  {
   \lc(x^{d_B-1}A) & \cdots & \cdots & \cdots & \cdots & \cdots & \cdots & \cc(x^{d_B-1}A)      \\
                   & \text{\llap{$\ddots$}}\hspace{0.5cm} &        &        &        &        &        & \vdots \\
                   &        & \text{\llap{$\lc(A)$}}\hspace{0.3cm} & \cdots & \cdots & \cdots & \cdots & \cc(A)      \\ 
   \lc(x^{d_A-1}B) & \cdots & \cdots & \cdots & \cdots & \cdots & \cdots & \cc(x^{d_A-1}B)      \\
                   & \text{\llap{$\ddots$}}\hspace{0.5cm} &        &        &        &        &        & \vdots \\
                   &        & \text{\llap{$\ddots$}}\hspace{0.3cm} &        &        &        &        & \vdots \\
                   &        &        & \text{\llap{$\lc(B)$}}\hspace{0.3cm} & \cdots & \cdots & \cdots & \cc(B)      \\ 
  };

\fill[fill=gray!50, opacity=0.3 ] (-2.3,0.5) -- (-5.2,2.5) -- (4.5,2.5) -- (4.5,0.5) -- cycle ;
\fill[fill=gray!50, opacity=0.3 ] (-1.6,-2.5) -- (-5.2,0.5) -- (4.5,0.5) -- (4.5,-2.5) -- cycle ;

\draw [decorate,decoration={brace,amplitude=6pt,mirror}]
(5.0,0.45) -- (5.0,2.5) node [black,midway,xshift=0.6cm] 
{\footnotesize $d_B$};

\draw [decorate,decoration={brace,amplitude=6pt,mirror}]
(5.0,-2.5) -- (5.0,0.35) node [black,midway,xshift=0.6cm] 
{\footnotesize $d_A$};

}
\end{tikzpicture}\\
\vspace{0.1cm}\scriptsize{Figure 1: The form of the Sylvester matrix of $A$ and~$B$. Entries outside of the gray area are zero.}
\end{center}
The Sylvester matrix $\syl(A,B)$ is defined to be the matrix of size $(d_A+d_B)\times (d_A+d_B)$ with the following entries: If $1\leq i\leq d_B$ and~$1\leq j\leq d_A+d_B$, the entry in the $i$th row and~$j$th column is the $(d_A+d_B-j)$th coefficient of $\D^{d_B-i}A$. If $d_B+1\leq i\leq d_A+d_B$ and~$1\leq j\leq d_A+d_B$, the entry in the $i$th row and~$j$th column is the $(d_A+d_B-j)$th coefficient of $\D^{d_A-(i-d_B)}B$. 

For $i,j\in\set N$ with $0\leq j\leq i \leq d_B$, the matrix $\syl_{i,j}(A,B)$ is obtained from $\syl(A,B)$ by removing the rows $1$ to $i$, the rows $d_B+1$ to $d_B+i$, the columns $1$ to $i$ and the last $i+1$ columns except for the column $d_A+d_B-j$. 
% Example for $i=2,j=1$:
\begin{center}
\begin{tikzpicture}
{
\matrix (m) [matrix of math nodes,left delimiter=(\hspace{0.7cm},right delimiter={)}]
  {
   \phantom{\lc(x^{d_B-1}A)} & \phantom{\Big(} & \phantom{\cdots} & \phantom{\cdots} & \phantom{\cdots} & \phantom{\cdots} & \phantom{\cdots} & \phantom{\cc(x^{d_B-1}A)}\hspace{0.5cm}\quad      \\
   \phantom{\lc(x^{d_B-1}A)}\\ 
  \phantom{\lc(x^{d_B-1}A)}\\
\phantom{\lc(x^{d_B-1}A)}\\
\phantom{\lc(x^{d_B-1}A)}\\
\phantom{\lc(x^{d_B-1}A)}\\
\phantom{\lc(x^{d_B-1}A)}\\
  };

\fill[fill=gray!50, opacity=0.3 ] (-2.3,0.5) -- (-5.2,2.5) -- (4.5,2.5) -- (4.5,0.5) -- cycle ;
\fill[fill=gray!50, opacity=0.3 ] (-1.6,-2.5) -- (-5.2,0.5) -- (4.5,0.5) -- (4.5,-2.5) -- cycle ;

\draw[fill=black] (-5.2,2.22) -- (4.6,2.22) ;
\draw[fill=black] (-5.2,1.72) -- (4.6,1.72) ;

\draw[fill=black] (-5.2,0.3) -- (4.6,0.3) ;
\draw[fill=black] (-5.2,-0.2) -- (4.6,-0.2) ;

\draw[fill=black] (-4.8,2.62) -- (-4.8,-2.5) ;
\draw[fill=black] (-4.2,2.62) -- (-4.2,-2.5) ;

\draw[fill=black] (4.25,2.62) -- (4.25,-2.6) ;
\draw[fill=black, decorate, decoration={border, segment length=10pt, angle=0} ] (3.65,2.6) -- (3.65,-2.6) ;
\draw[fill=black] (2.95,2.62) -- (2.95,-2.6) ;
}
\end{tikzpicture}\\
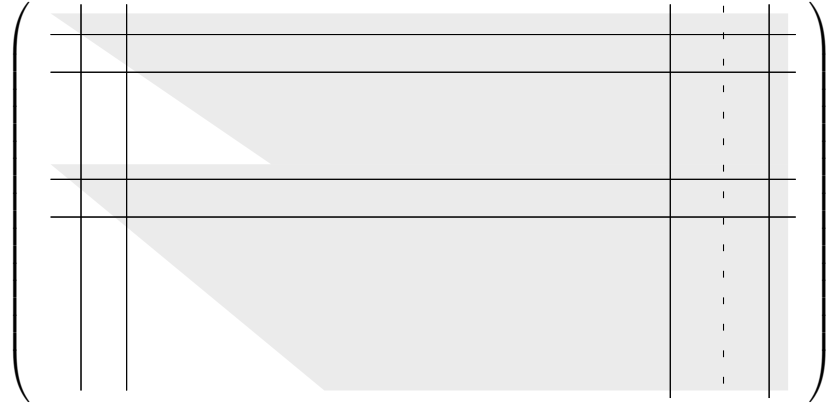
\vspace{0.1cm}\scriptsize{Figure 2: Sketch of $\syl_{2,1}(A,B)$. The lines indicate the removed rows and columns. The column under the dotted line is added again.}
\end{center}
\begin{definition}
 For $0\leq i\leq d_B$, the polynomial \[\pres_i(A,B) := \sum_{j=0}^{i}{\det(\syl_{i,j}(A,B))\D^j}\] is called the $i$th (polynomial) subresultant of $A$ and~$B$.
 If the order of $\pres_i(A,B)$ is strictly less than $i$, the $i$th subresultant of $A$ and~$B$ is called defective, otherwise it is called regular.
 The subresultant sequence of $A$ and~$B$ of the first kind is the subsequence of \[(A,B,\pres_{d_B-1}(A,B),\pres_{d_B-2}(A,B),\dots,\pres_{0}(A,B),0)\] that contains $A$, $B$, the trailing zero and all nonzero $\pres_{i}(A,B)$ for which $\pres_{i+1}(A,B)$ is regular.
\end{definition}

\begin{theorem}[\cite{zli}]
\label{zlithm}
 The polynomial remainder sequence given by $\alpha_i$ and~$\beta_i$ as in Example~\ref{prsex}.\ref{sresprs}, the subresultant PRS, is equal to the subresultant sequence of $A$ and~$B$ of the first kind.
\end{theorem}

\section{Identifying Content of Polynomial Subresultants}
\label{contentsec}

The representation of subresultants in terms of determinants of the matrices $\syl_{i,j}(A,B)$ makes it possible to identify content by exploiting the special form of these matrices as well as the correspondence between rows of the Sylvester matrix and monomial multiples of $A$ and~$B$. For the case of commutative polynomials, some results are known for detecting such additional content. We generalize two results to the Ore setting. The first (Theorem~\ref{simpleimprov}) is a generalization of an observation mentioned in \cite{brown}, which carries over quite easily to the Ore case. The second (Theorem~\ref{mainthm}) usually performs better in terms of coefficient size of the remainders, but a heuristic argument is necessary to use it algorithmically (see Section~\ref{improvsec}). 

\begin{theorem}
\label{simpleimprov}
 With $t:= \gcd(\sigma^{d_B-1}(\lc(A)),\sigma^{d_A-1}(\lc(B)))$ and~$\gamma_{i} := \sigma^{-i}(t)$ for $0\leq i\leq d_B-1$, we get: $\gamma_i\mid\cont(\pres_i(A,B)).$
\end{theorem}
\begin{proof}
 Let $i$ be fixed. The coefficients of $\pres_i(A,B)$ are the determinants of the matrices $\syl_{i,j}(A,B)$ for $0\leq j\leq i$. The first column of all of these matrices is
\[(\sigma^{d_B-1-i}(\lc(A)),0,\dots,0,\sigma^{d_A-1-i}(\lc(B)),0,\dots,0)^\mathsf{T}.\]
Laplace expansion along this column proves the claim.\qed
\end{proof}

Not all of the subresultants of $A$ and~$B$ are in the subresultant PRS of $A$ and~$B$. To make use of Theorem~\ref{simpleimprov} for a new PRS, we need a minor specialisation of the statement:

\begin{corollary}
\label{cor1}
  Let $\prs{R}{0}{\ell+1}$ be the subresultant PRS of $A$ and~$B$ (not necessarily normal). If we choose \[t=\gcd(\sigma^{d_B-1}(\lc(A)),\sigma^{d_A-1}(\lc(B))),\quad \gamma_2 = \sigma^{-d_B+1}(t)\text{ and }\gamma_i=\sigma^{d_{i-2}-d_{i-1}}(\gamma_{i-1})\text{ for } 2<i\leq\ell,\] then $\gamma_i\mid\cont(R_i)$ for $2\leq i\leq\ell$.
\end{corollary}
\begin{proof}
 Suppose $R_i$ is the $j$th subresultant of $A$ and~$B$. Then, by the definition of the subresultant sequence of the first kind and Theorem~\ref{zlithm}, the $(j+1)$st subresultant of $A$ and~$B$ is regular. Because of this and the subresultant block structure (see \cite{zli}), $R_{i-1}$ is of order $j+1$ and so $j$ is equal to $d_{i-1}-1$. By Theorem~\ref{simpleimprov}, the content of $R_i$ is divisible by $\sigma^{-d_{i-1}+1}(t)$. 
 It is easy to see that $\sigma^{-d_{i-1}+1}(t)$ is equal to $\gamma_i$.\qed
\end{proof}

In the commutative case, a second source of additional content was determined, although this result is not widely known. The following theorem can be found in \cite{knuth}:

\begin{theorem}
\label{knuththm}
 Let $A,B\in \k[\D]$ be such that the subresultant PRS of $A$ and~$B$ is normal, i.e. $d_{i-1} = d_i+1$ for $1\leq i \leq \ell$, and let $G$ be the GCD of $A$ and~$B$. Then $\lc(G)^{2(i-1)}\mid \cont(R_i)$ for $2\leq i\leq \ell$.
\end{theorem}

A generalization of Theorem~\ref{knuththm} to Ore polynomials is not straightforward, as Example~\ref{knuthex} shows.

\begin{example}[Example~\ref{mainex} cont.]
 \label{knuthex}
If we take $A$ and~$B$ as in Example~\ref{mainex}, then the leading coefficient of the GCRD of $A$ and~$B$ is $(n+9)p(n)$,
where $p(n)$ is a polynomial of degree 17. The subresultant PRS of $A$ and~$B$ turns out to be normal and~$R_2$ is of order $d_2 = 12$. By Theorem~\ref{knuththm}, if the polynomials were elements of $\k[\D]$, $\cont(R_2)$ would be divisible by $\lc(G)^2$ and a naive translation of the theorem to the non-commutative case suggests divisibility by a polynomial of degree at least~36. The (monic) content of $R_2$, however, is only $(n+16)(n+17) $, which is contained in, but not equal to, $\sigma^{7}(\lc(G))^{[2]}$.
\end{example}

Again in the commutative case, let $Q_A$, $Q_B \in\k[\D]$ be such that $A=Q_AG$ and $B=Q_BG$. \cite{knuth} proves Theorem~\ref{knuththm} by showing that if $\prs{R}{0}{\ell+1}$ is the subresultant PRS~of~$A$~and~$B$ and~$\prs{\tilde{R}}{0}{\ell+1}$ is the subresultant PRS of $Q_A$, $Q_B$, then $Q_i = \lc(G)^{2(i-1)}\tilde{R}_i$. This approach is problematic for Ore polynomials, because there the $Q_i$'s and the $\tilde{R}_i$'s have coefficients in $\qk$ and not necessarily in $\k$. This means that even after showing that a quotient $Q_i$~is~a $\k$-left multiple of some subresultant $\tilde{R}_i$ of $Q_A$ and~$Q_B$, the left factor and the denominators in the coefficients of $\tilde{R}_i$ might not be coprime and thus lead to cancellation. Therefore we will not only describe why in the non-commutative case only some factors of $\lc(G)$ appear as content, but we also present a new proof of Theorem~\ref{knuththm} that makes it more explicit where the additional content comes from. Moreover, we won't require the 
remainder sequence to be normal. 

In $\k[x]$, if $A$ is a multiple of the primitive polynomial $G$, then their quotient will always have coefficients in $\k$, and therefore, the leading coefficient of $A$ contains all the factors of the leading coefficient of $G$. For Ore polynomials, this is not necessarily true, since the quotient of $A$ and $G$ might be an element of $\qalg\setminus\alg$. Still, different left multiples of $G$ in $\alg$ may share some common factors in their leading coefficients, as described in Lemma \ref{lcthm}.

\begin{lemma}
\label{lcthm}
 Let $d_T\in\set N$ be fixed, let $\mathcal{I}\vartriangleleft\alg$  be a left ideal and let $T$ be any element of $\mathcal{I}$ of order $d_T$ such that, among all the operators of order $d_T$ in $\mathcal{I}$, its leading coefficient~$t$ is minimal with respect to the degree. Then $t$ is independent of the choice of $T$ (up to multiplication by units in $\set D$) and for any $L\in \mathcal{I}$ with $d_L\leq d_T$ we have $\sigma^{d_L-d_T}(t)\mid\lc(L)$. 
\end{lemma}
\begin{proof}
 Assume there are $T,L\in \mathcal{I}$ for which the claim $\sigma^{d_L-d_T}(t)\mid\lc(L)$ does not hold.
 We let $L' =\D^{d_T-d_L}L$ and get $\lc(L') = \sigma^{d_T-d_L}(\lc(L))$, thus $t\nmid\lc(L')$ by assumption.
 Division with remainder yields nonzero $q,r\in \k$ such that \[\lc(L') = qt+r,\quad \deg(r)<\deg(t).\]
 Hence the operator $L'-qT$ is an element of $\mathcal{I}$ whose leading coefficient has degree less than $\deg(t)$. This contradicts the choice of $T$.

\noindent For the uniqueness, let $T'\in \mathcal{I}$ be any other operator of order $d_T$ with minimal leading \mbox{coefficient} degree. By what was just shown above, we get $\lc(T')\mid t$ and $t\mid\lc(T')$, so $t$ and $\lc(T')$ are associates.
\qed
\end{proof}

 \begin{definition}
  Consider $\mathcal{I}$, $T$ and $t$ from Lemma~\ref{lcthm}. The shift $\sigma^{-d_T}(t)$ of the leading coefficient of~$T$ is called the essential part of $\mathcal{I}$ at order~$d_T$. If there is no operator in $\mathcal{I}$ for some order $n$, the essential part of $\mathcal{I}$ at order $n$ is defined to be $1$.
 \end{definition}

Let $L\in\set C[y][D;1,\frac{d}{dy}]$ and $\mathcal{I}=\mathcal{I}'\cap\set C[y][D;1,\frac{d}{dy}]$ where $\mathcal{I}'\vartriangleleft\set  C(y)[D;1,\frac{d}{dy}]$ is the left ideal generated by~$L$. We give an informal explanation of essential parts of $\mathcal{I}$ in terms of solutions of~$L$, i.e. functions that are annihilated by $L$. Any non-removable singularity of a solution of $L$ corresponds to a root of the leading coefficient of $L$, but not for any root of $\lc(L)$ there has to be a solution with a non-removable singularity at that point. Any solution of $L$ is also a solution of every operator in~$\mathcal{I}$ and it can happen that there are nonzero $\qk$-left multiples of $L$ in $\mathcal{I}$ that have strictly smaller leading coefficient degree than $L$. If such a \textit{desingularized} operator exists, it means that some of the roots of $\lc(L)$ can be removed by multiplying $L$ with another operator from the left. These removable roots are called the \textit{apparent singularities} of 
$L$. It is shown in \cite{mythesis} that there exists a unique minimal (w.r.t.\ degree) essential part of $\mathcal{I}$ that appears in the essential parts of $\mathcal{I}$ at every order greater than $d_L$. This minimal essential part of $\mathcal{I}$ is a polynomial whose roots are exactly the non-apparent singularities of $L$, and it turns out that for each root of the essential part of $\mathcal{I}$, there is at least one solution of $L$ that does not admit an analytic continuation at that point. A more detailed description of desingularization and apparent singularities of differential equations can be found in \cite{ince}. Further references and recent results on desingularization of Ore operators can be found in \cite{chen}.

Note that for commutative polynomials, by Gau\ss' Lemma, the essential part of a nonzero ideal at any order is equal to the leading coefficient of the primitive greatest common divisor of the ideal elements. 

For the remaining part of this article, let $\mathcal{I}\vartriangleleft\alg$ be the left ideal generaed by $A$ and~$B$. We formulate our Ore generalization of Theorem~\ref{knuththm}, where now some of the essential parts of $\mathcal{I}$ play the role of the leading coefficient of the GCRD of $A$ and $B$.

\begin{theorem}
\label{mainthm}
 Let $i\in\{0,\dots,d_B-1\}$ and~$\Delta:=d_A+d_B-2i$. If $t_k$ is the essential part of $\mathcal{I}$ at order~$k$ for $i< k\leq \Delta+i-1$, then
\[\left(\prod_{k=i+1}^{\Delta+i-1}t_k\right)\mid\cont(\pres_i(A,B)).\]
\end{theorem}

\begin{proof}
 For any $j\in\{0,\dots,i\}$, $\syl_{i,j}(A,B)$ is of size $\Delta\times\Delta $ and if the last column is removed, the resulting matrix does not depend on $j$ anymore. For $n\in \{1,\dots,\Delta-1\}$, let $\mathcal{M}_{i,n}$ be the set of all $n\times n$ matrices obtained by removing the last $\Delta-n$ columns and any $\Delta-n$ rows from $\syl_{i,j}(A,B)$. The $j$th coefficient of $\pres_i(A,B)$ is the determinant of $\syl_{i,j}(A,B)$ and Laplace expansion along the last column shows that it is a $\k$-linear combination of the elements of $\mathcal{M}_{i,\Delta-1}$. By induction on~$n$ we show that the determinant of any element of $\mathcal{M}_{i,n}$ is divisible by $t_{\Delta+i-n}t_{\Delta+i-(n-1)}\dots  t_{\Delta+i-1}$. The theorem is then proven by setting $n=\Delta-1$.

 For $n=1$, the only entry in a matrix in $\mathcal{M}_{i,1}$ is either zero or the leading coefficient of a monomial left multiple of $A$ or $B$ of order $\Delta+i-1$, so the claim follows from Lemma~\ref{lcthm}.

 Now suppose the claim is true for $1\leq n < \Delta-1$ and let $M$ be any element of $\mathcal{M}_{i,n+1}$. If the determinant of $M$ is zero, then there is nothing to show. Consider the case where $\det(M)\neq 0$. 
Then there is a $v\in\qk^{n+1}$ such that $M^\mathsf{T}v=(0,\dots,0,1)^\mathsf{T}$. By Cramer's rule, the $j$th component~$v_j$ of $v$ is of the form $p_j/\det(M)$ where $p_j\in\k$ is the determinant of some element of $\mathcal{M}_{i,n}$. By induction hypothesis it is divisible by $t_{\Delta+i-n}t_{\Delta+i-(n-1)}\dots t_{\Delta+i-1}$. 
 Every row in $M$ corresponds to an operator of the form $\D^kA$ or $\D^kB$ for $k\in\set N$, minus some of the lower order terms. For the $j$th row, $1\leq j\leq n+1$, we denote the corresponding operator by $L_j$. By the definition of $v$, the operator $\sum_{j=0}^{n+1}{v_jL_j}\in\qalg$ will have order $\Delta+i-(n+1)$ and leading coefficient $1$.
 So if we set \[v':=\frac{\det(M)}{t_{\Delta+i-n}t_{\Delta+i-(n-1)}\dots t_{\Delta+i-1}}v \in\k^{n+1}\] and~$L=\sum_{j=0}^{n+1}{v'_jL_j}$, then $L$ is an element in $\mathcal{I}$ of order $\Delta+i-(n+1)$ and its leading coefficient is $\det(M)/(t_{\Delta+i-n}t_{\Delta+i-(n-1)}\dots t_{\Delta+i-1})\in\k$. Lemma~\ref{lcthm} yields that $\lc(L)$ is divisible by $t_{\Delta+i-(n+1)}$, so we get in total $t_{\Delta+i-(n+1)}t_{\Delta+i-n}\dots t_{\Delta+i-1}\mid\det(M)$.\qed
\end{proof}

In practice, the essential parts of $\mathcal{I}$ will most likely be the same at every order $n$ with $d_G\leq n\leq d_A+d_B$. In that case, Theorem~\ref{mainthm} is equivalent to the following simplification, where only the essential part of $\mathcal{I}$ at order $d_A+d_B$ needs to be known.

\begin{corollary}
\label{oneespart}
 Let $i\in\{0,\dots,d_B-1\}$ and~$\Delta:=d_A+d_B-2i$. If $t$ is the essential part of $\mathcal{I}$ at order~$d_A+d_B$, then
\[\sigma^{i+1}(t)^{[\Delta-1]}\mid\cont(\pres_i(A,B)).\]
\end{corollary}

\begin{proof}
 According to Lemma~\ref{lcthm}, $\sigma^{j}(t)$ divides the essential part of $\mathcal{I}$ at order $j$ for any $d_G\leq j\leq d_A+d_B$. If $i<d_G$, then the $i$th subresultnat of $A$ and $B$ is zero. Otherwise, Theorem~\ref{mainthm} yields that $\cont(\pres_i(A,B))$ is divisible by 
 \[\sigma^{i+1}(t)\sigma^{i+1}(t)\dots \sigma^{\Delta+i-1}(t) = \sigma^{i+1}(t)^{[\Delta-1]}.\text{\rlap{$\hspace{4.3cm}\square$}}\]
\end{proof}

Like for Theorem~\ref{simpleimprov}, an adjustment of Corollary~\ref{oneespart} to the block structure of the subresultant sequence of the first kind is needed in order to construct a new PRS.

\begin{corollary}
\label{cor2}
 Let $\prs{R}{0}{\ell+1}$ be the subresultant PRS of $A$ and~$B$ (not necessarily normal) and let $t$ be the essential part of $\mathcal{I}$ at order $d_A+d_B$. If we set $\gamma_2 = \sigma^{d_B}(t)^{[d_A-d_B+1]}$ and \[\gamma_i=\sigma^{d_{i-1}}(t)^{[d_{i-2}-d_{i-1}]}\gamma_{i-1}\sigma^{d_A+d_B-d_{i-2}+1}(t)^{[d_{i-2}-d_{i-1}]}\text{ for } 2<i\leq\ell,\] then $\gamma_i\mid\cont(R_i)$ for $2\leq i\leq\ell$.
\end{corollary}
\begin{proof}
 Suppose $R_i$ is the $j$th subresultant of $A$ and~$B$. As in the proof of Corollary~\ref{cor1}, we have that~$j$ is equal to $d_{i-1}-1$. So by Corollary~\ref{oneespart}, the content of $R_i$ is divisible by $\sigma^{d_{i-1}}(t)^{[d_A+d_B-2d_{i-1}+1]}$. Simple hand calculation shows that this is equal to $\gamma_i$.\qed
\end{proof}

\section{Improved Polynomial Remainder Sequence}
\label{improvsec}

We now derive formulas for the $\alpha_i$ and~$\beta_i$ that take into account the potential additional content characterized by Theorems~\ref{simpleimprov} and~\ref{mainthm}. For this we need the following lemma:

\begin{lemma}
\label{quothm}
 For $\gamma_1,\gamma_2\in\qk$:
 $\pquo(\gamma_1A,\gamma_2B)\gamma_2 = \gamma^{\mathstrut}_1\gamma_2^{[d_A-d_B+1]}\pquo(A,B).$
\end{lemma}
\begin{proof}
 By Lemma 2.3 in \cite{zli}, the pseudo-remainder of $\gamma_1A$ and~$\gamma_2B$ is the $(d_B-1)$st subresultant of $\gamma_1A$ and~$\gamma_2B$ (up to sign). Consequently, its coefficients are determinants of submatrices of $\syl(\gamma_1A,\gamma_2B)$ that contain one row corresponding to the operator $\gamma_1A$ and~$d_A-d_B+1$ rows corresponding to operators of the form $\D^i\gamma_2B$, $0\leq i\leq d_A-d_B$. Thus, by Lemma~2.2 in~\cite{zli}, it follows that (up to sign)
\begin{equation}
\label{eq4}
\prem(\gamma_1A,\gamma_2B) = \gamma^{\mathstrut}_1\gamma_2^{[d_A-d_B+1]}\prem(A,B).
\end{equation} 
The pseudo-remainder formula \eqref{premformula} applied to  $\gamma_1 A$ and~$\gamma_2 B$ is\[\lc(\gamma_2B)^{[d_A-d_B+1]}\gamma_1A=\pquo(\gamma_1A,\gamma_2B)\gamma_2B + \prem(\gamma_1A,\gamma_2B).\]
Combining this with \eqref{eq4} and dividing the resulting equation by $ \gamma^{\mathstrut}_1\gamma_2^{[d_A-d_B+1]}$ from the left gives the desired result. \qed
\end{proof}

This now allows us to state $\alpha_i$ and~$\beta_i$ for improved polynomial remainder sequences:

\begin{theorem}
\label{newabs}
 Suppose $\prs{R}{0}{\ell+1}$ is the subresultant PRS of $A$ and~$B$ and~$\prs{\gamma}{0}{\ell+1}$ is any sequence in $\qk\setminus\{0\}$ with $\gamma_0=\gamma_1=1$. Set $\tilde{R}_i = \frac{1}{\gamma_i}R_i$. Then $\prs{\tilde{R}}{0}{\ell+1}$ is a PRS of $A$ and~$B$ with: 
\begin{alignat*}2
 \tilde{\alpha}_i & = \lc(\tilde{R}_i)^{[d_{i-1}-d_i+1]},\\
 \tilde{\beta}_i & = \left\{\begin{array}{ll}
\displaystyle -\sigma(\tilde{\psi}_1)^{[d_0 - d_1]}\gamma_2, & \text{ if } i=1,\\
\displaystyle \frac{-\lc(\tilde{R}_{i-1})\sigma(\tilde{\psi}_i)^{[d_{i-1}-d_i]}}{{\gamma_i}^{[d_{i-1}-d_i+1]}}\gamma_{i+1}, & \text{ if } 2\leq i\leq\ell,\end{array}\right.
\end{alignat*}
where
\[
\tilde{\psi}_i = \left\{\begin{array}{ll}
                 \displaystyle -1, & \text{ if } i = 1,\\
                 \displaystyle \frac{(-\gamma_{i-1}\lc(\tilde{R}_{i-1}))^{[d_{i-2}-d_{i-1}]}}{\sigma(\tilde{\psi}_{i-1})^{[d_{i-2}-d_{i-1}-1]}}, & \text{ if } 2\leq i\leq\ell.
                \end{array}\right.
\]
\end{theorem}

\begin{proof}
 From the definition of $\tilde{R}_i$ and the equations
\begin{equation*}
 \alpha_i R_{i-1} = Q_iR_i + \beta_i R_{i+1}\text{\qquad and \qquad}\alpha_i = \gamma_i^{[d_{i-1}-d_i+1]}\tilde{\alpha}_i,
\end{equation*}
it follows that
\begin{equation}
\label{eq1}
 \gamma_i^{[d_{i-1}-d_i+1]}\gamma_{i-1}\tilde{\alpha}_i \tilde{R}_{i-1} = {Q}_i\gamma_i\tilde{R}_i +  {\beta}_i\gamma_{i+1}\tilde{R}_{i+1}.
\end{equation}
For the first summand on the right hand side, Lemma~\ref{quothm} yields
\begin{equation}
\label{eq2}
 {Q}_i\gamma_i = \gamma_i^{[d_{i-1}-d_i+1]}\gamma_{i-1}\tilde{Q}_i.
\end{equation}
For the second summand, observe that since $\gamma_i\lc(\tilde{R}_i)$ equals $\lc(R_i)$, we have that $\psi_i$ equals $\tilde{\psi}_i$ for all $1\leq i\leq\ell$. Thus
\begin{equation}
\label{eq3}
 {\beta}_i\gamma_{i+1} = \gamma_i^{[d_{i-1}-d_i+1]}\gamma_{i-1}\tilde{\beta}_i.
\end{equation}
The proof is concluded by combining \eqref{eq1}, \eqref{eq2} and \eqref{eq3} and dividing the resulting equation by $\gamma_i^{[d_{i-1}-d_i+1]}\gamma_{i-1}$ from the left.\qed
\end{proof}

Two possible choices for $\prs{\gamma}{i}{\ell+1}$ were presented in Corollary~\ref{cor1} and \ref{cor2}. The computation of~$\gamma_i$ in Corollary~\ref{cor1} is straightforward, but in Corollary~\ref{cor2}, the essential part of $\mathcal{I}$ (the ideal generated by $A$ and $B$) at order~$d_A+d_B$ is usually not known. A simple heuristic can solve this problem in most cases: As was shown in Lemma~\ref{lcthm}, the essential part of $\mathcal{I}$ at order~$d_A+d_B$ appears in a shifted version in the leading coefficient of every nonzero ideal element with order less than or equal to $d_A+d_B$. In particular it is contained in $\lc(A)$ and~$\lc(B)$. Thus, if $t$ is the essential part of~$\mathcal{I}$ at order $d_A+d_B$, we have
\begin{equation}
\label{eq5}
\sigma^{d_A}(t) \mid \gcd(\lc(A),\sigma^{d_A-d_B}(\lc(B))) 
\end{equation}
and in most cases, we not only have divisibility but equality. In fact, in all the examples we looked at that came from combinatorics or physics, this guess for the essential part turned out to be correct.

\begin{example}[Example~\ref{knuthex} cont.]
 We now use Theorem~\ref{newabs} and Corollaries~\ref{cor1}~and~\ref{cor2} to compute new PRSs of $A$ and~$B$ as in Example~\ref{mainex}. The essential part of $\mathcal{I}$ at order~$d_A+d_B$ is $(n+3)$, so $\sigma^{d_A}(n+3) = (n+17)$, which is also the guess given by the right hand side of \eqref{eq5}. Applying Corollary \ref{cor1} yields the factors
\[\gamma_2 = n+17,\quad \gamma_3 = n+18,\quad\dots\quad \gamma_i = n+16+i-1,\quad\dots\]
whereas Corollary \ref{cor2} gives
\[\gamma_2 = (n + 16)^{[2]},\quad \gamma_3 = (n + 15)^{[4]},\quad \dots\quad \gamma_i = (n+16-i+2)^{[2(i-1)]},\quad\dots\]
 The improvements from Corollary~\ref{cor1}  are marginal, while the degrees in the improved PRS with the results from Corollary~\ref{cor2} are equal to the degrees in the primitive PRS, except for the very last step:
 \begin{center}
\begin{tabular}{l|c|c|c|c|c|c|c}
% \hline
PRS & $R_2$ & $R_3$ & $R_4$ & $R_5$ & $R_6$ & $R_7$ & $R_8$\\
\hline
 subresultant & 11 & 16 & 21 & 26 & 31 & 36 & 41\\
\hline
 improved (Cor.~\ref{cor1}) & 10 & 15 & 20 & 25 & 30 & 35 & 40\\
\hline
 improved (Cor.~\ref{cor2})& 9 & 12 & 15 & 18 & 21 & 24 & 27\\
\hline
 primitive & 9 & 12 & 15 & 18 & 21 & 24 & 21\\
\end{tabular}\\
\vspace{0.1cm}\scriptsize{Table 2: Maximal coefficient degrees for the subresultant, improved and primitive PRS.}
\end{center}
\end{example}

\begin{example}
 Although the remainders in the PRS based on Corollary~\ref{cor2} are usually primitive when starting from randomly generated operators or operators that come from some applications, it is not guaranteed that this is always the case. As an example, consider 
\begin{alignat*}2
 \text{\rlap{$A,B\in\mathbb{Q}[y][x],$}}\ \ & \\
 A &= x^4+yx^2+yx+y,\\
 B &= x^3+yx^2.
\end{alignat*}
The second subresultant of $A$ and~$B$ is $\pres_2(A,B)=(y+y^2)x^2+yx+y$, so $\cont(\pres_2(A,B))=y$, but in the improved PRS, no content will be found. 

As mentioned, it may also happen that the guess for the essential part of $\mathcal{I}$ at order $d_A+d_B$ is too large, for example:
\begin{alignat*}2
 \text{\rlap{$A,B\in\mathbb{Q}(y)[D,1,\tfrac{d}{dx}],$}}\ \ & \\
 A&=(y + 1)D^4 + D^3 + D^2 + yD + 1,\\
 B&=(y + 1)D^3 + D^2 + 1.
\end{alignat*}
Here, $\operatorname{cont}(R_3$) in the subresultant PRS is $(y+1)$, but a factor $(y+1)^2$ is predicted. The mistake in predicting the essential part can be noticed on the fly during the execution of the algorithm as soon as a remainder with coefficients in $\mathbb{Q}(y)$ appears. It is then possible to either switch to another PRS or to refine the guess of the essential part. One strategy to do so is to remove all the factors from the guess that could be responsible for the appearance of denominators. Let $t$ be the guess for the essential part of $\mathcal{I}$ at $d_A+d_B$ and let $c$ be the non-trivial common denominator of the coefficients of a remainder $R_i$ in the improved PRS. Furthermore let $M$ be the set of all integers~$m$ such that $\gcd(\sigma^m(c),t)\neq 1$. Update $R_i$, $\gamma_i$ and $t$ with
\begin{alignat*}2
 & R_i & \leftarrow &\ cR_i,\\
 & \gamma_i & \leftarrow &\ \frac{\gamma_i}{c},\\
 & t & \leftarrow &\ \frac{t}{\gcd(t,\prod_{m\in M}{\sigma^m(c)})},\\
 & \gamma_{i+1}\quad & \leftarrow &\ \sigma^{d_{i}-d_B}(t)^{[d_A+d_B-2d_{i}+1]},
\end{alignat*}
and continue the computation with these new values. For differential operators in $\set C(y)[D;1,\frac{d}{dy}]$, we have $M=\{0\}$ and for recurrence operators in $\set C(n)[S_n;s_n,0]$, $M$ contains all the integer roots of the polynomial $\operatorname{res}_n(c(n+m),t)$.
\end{example}

\begin{example}
 We can guess two operators $A$ and $B$ in $\set Q[n][S;s_n,0]$ of order $d_A=16$, $d_B=14$, resp. that annihilate the sequence
\[t_n = (7n^3+5n^2+n+1)^7({(n+1/7)^{\overline{12}})}^7\frac{(2n)!^3}{(3n)!^2}.\]
The GCRD of $A$ and $B$ is of order $1$ and the essential part of $\mathcal{I}$ at $d_A+d_B$ is of degree $4$. The essential part of $\mathcal{I}$ at order $11$, however, is of degree $11$, so here we are in the rare case where the essential part of $\mathcal{I}$ at order $d_A+d_B$ is only contained but not equal to the essential part at lower orders. Formula \eqref{eq5} only predicts the essential part of $\mathcal{I}$ at order  $d_A+d_B$ and during the GCRD computation, content that comes from lower order essential parts emerges
 \begin{center}
\begin{tabular}{l|c|c|c|c|c|c|c}
% \hline
PRS & $R_2$ & $R_3$ & $R_4$ & $R_5$ & $R_6$ & $R_7$ & $R_8$\\
\hline
 improved (Cor.~\ref{cor2})& 31 & 44 & 57 & 70 & 83 & 96 & 109 \rlap{\quad}\\
\hline 
 primitive & 31 & 44 & 50 & 56 & 62 & 68 & 74\\
\end{tabular}\\
{\vspace{0.1cm}\scriptsize{Table 3: Maximal coefficient degrees for the first few remainders in the improved and primitive PRS.}}
\end{center}
It is possible to guess the essential part of $\mathcal{I}$ at lower orders and then use Theorem~\ref{mainthm} to get the primitive remainders, but like in the direct computation of the primitive PRS, GCD computations in the base ring would be necessary after each division step.
\end{example}

 As another consequence of Theorem~\ref{mainthm}, we can give a new bound for the coefficient degrees of the primitive PRS in terms of the essential parts of the left ideal generated by $A$ and $B$. 
 \begin{theorem}
  Let $\prs{R}{0}{\ell+1}$ be the primitive PRS of $A$ and~$B$. Fix $i\in\{0,\dots,\ell\}$ and let $b\in\set N$ be such that $\max_{k\in\{0,\dots,d_B-d_{i-1}\}}(\lVert\D^kA\rVert)\leq b$ and $\max_{k\in\{0,\dots,d_A-d_{i-1}\}}(\lVert \D^kB\rVert)\leq b$. If $t_j$ denotes the essential part of $\mathcal{I}$ at order $j\in\set N$, then
  \begin{alignat*}2
   \lVert R_i\rVert & \leq (d_A+d_B-2(d_{i-1}-1))b - \sum_{j=d_{i-1}}^{\mathclap{d_A+d_B-d_{i-1}+1}}{\deg(t_j)}
%                     & \leq (d_A+d_B-2(d_{i-1}-1))b - (d_A+d_B-2d_{i-1}+1)\deg(t_{d_A+d_B}).
  \end{alignat*}
 \end{theorem}
\begin{proof}
 The bound follows directly from Hadamard's inequality, the subresultant block structure and Corollary~\ref{cor2}. \qed
\end{proof}

\section*{Acknowledgements}
I would like to thank Ziming Li and Manuel Kauers for their helpful comments and support during our personal communication. Also, I thank the reviewers for their careful reading of this article and for pointing out mistakes and shortcomings in the draft.

\bibliographystyle{model2-names}
\bibliography{main}
\end{document}